\documentclass[10pt]{article}
\usepackage{amsmath,amssymb,amsfonts}
\usepackage{amsthm}
\usepackage{graphicx}
\usepackage{color}
\usepackage{subfigure}
\usepackage[hypertex]{hyperref}
\usepackage{mathpazo}

\newcommand{\bra}[1]{\ensuremath{\langle #1 |}}
\newcommand{\ket}[1]{\ensuremath{| #1 \rangle}}

\newcommand*{\cH}{\mathcal{H}}

\newcommand{\pmset}[1]{\{\pm1\}^{#1}}
\newcommand{\psiket}{\ket{\psi}}
\newcommand{\psibra}{\bra{\psi}}

\newcommand{\beq}{\begin{equation}}
\newcommand{\eeq}{\end{equation}}
\newcommand{\beqn}{\begin{equation*}}
\newcommand{\eeqn}{\end{equation*}}
\newcommand{\beqrn}{\begin{eqnarray*}}
\newcommand{\eeqrn}{\end{eqnarray*}}

\newcommand{\CC}{\mathbb{C}}
\newcommand{\RR}{\mathbb{R}}

\def\one{\leavevmode\hbox{\small1\normalsize\kern-.33em1}}

\newtheorem{theorem}{Theorem}
\newtheorem{lemma}[theorem]{Lemma}

\newtheorem{proposition}[theorem]{Proposition}

\theoremstyle{definition}
\newtheorem{definition}[theorem]{Definition}

\renewcommand{\qed}{\hfill{\rule{2mm}{2mm}}}
\renewenvironment{proof}[1][]{\begin{trivlist}
\item[\hspace{\labelsep}{\bf\noindent Proof#1:\/}] }{\qed\end{trivlist}}

\renewcommand\vec[1]{\ensuremath{#1}}
\newcommand\vsequence[2]{\ensuremath{{\vec #1_1,\ldots,\vec #1_{#2}}}}
\newcommand\ssequence[2]{\ensuremath{{{#1}_1,\ldots,{#1}_{#2}}}}
\newcommand\kgnm{\ensuremath{K_G(n\mapsto m)}}
\newcommand\eps\varepsilon

\begin{document}

\title{A generalized Grothendieck inequality and entanglement in XOR games} \author{ Jop
  Bri\"et\footnote{Centrum voor Wiskunde en Informatica, Kruislaan
    413, 1098 SJ Amsterdam, The Netherlands. Supported by Vici grant
    639-023-302 from the Netherlands Organization for Scientific
    Research (NWO), by the European Commission under the Integrated
    Project Qubit Applications (QAP) funded by the IST directorate as
    Contract Number 015848, and by the Dutch BSIK/BRICKS
    project. BT is also supported by BQP Solutions Pty Ltd. Email: jop.briet@cwi.nl, buhrman@cwi.nl, bentoner@bentoner.com.} \and Harry Buhrman$^*$ \and Ben
  Toner$^*$}

\maketitle

\begin{abstract}Suppose Alice and Bob make local two-outcome
  measurements on a shared entangled state. For any~$d$, we show that
  there are correlations that can only be reproduced if the local
  dimension is at least~$d$. This resolves a conjecture of Brunner et
  al. [\textit{Phys. Rev. Lett.}  100, 210503 (2008)] and establishes
  that the amount of entanglement required to maximally violate a Bell
  inequality must depend on the number of measurement settings, not
  just the number of measurement outcomes. We prove this result by
  establishing the first lower bounds on a new generalization of
  Grothendieck's constant.
\end{abstract}

\section{Introduction}
Grothendieck's inequality first arose in the study of norms on tensor
products of Banach spaces~\cite{grothendieck53:_resum}. It has since
found many applications in mathematics and computer science, including
approximation algorithms~\cite{alon04:_approx,Charikar2004} and
communication complexity~\cite{Linial2007,Linial2008}. In quantum information, it
quantifies the difference between the classical and quantum values of
certain simple Bell inequalities, as established by
Tsirelson~\cite{Tsirelson:85b}. Tsirelson's work has been the
starting point for considerable recent 
research into quantum
nonlocality~\cite{brunner:210503,toner:groth1,Cleve:04a,regev:_simul}.

We start by stating the
inequality in its strongest form, in terms of \emph{the real Grothendieck
  constant} $K_G$. 

\begin{definition}
The \emph{real Grothendieck constant of order $n$}, is the smallest
real number $K_G(n)$ such that: For all positive
integers $r$ and for all 
real $r
\times r$ 
matrices $M = (M_{ij})$, 
the inequality 
\begin{align}\label{gineq}
  \max_{\substack{\vsequence{a}{r}\\\vsequence{b}{r}}}\,
  \sum_{i,j}M_{ij} 
\vec a_i \cdot \vec b_j \leq
K_G(n)\max_{\substack{\ssequence{\alpha}{r}\\\ssequence{\beta}{r}}}\,\sum_{i,j}M_{ij}\alpha_i
\beta_j  \end{align}
holds, where the maximum on the left-hand side is taken over all 
sequences $\vsequence{a}{r},\vsequence{b}{r}$ of $n$-dimensional real unit
vectors, $a_i \cdot b_j$ denotes the Euclidean inner product of
$a_i$ and $b_j$, and the maximum on the right-hand side is taken over
all for all sequences 
$\vsequence{\alpha}{r},\vsequence{\beta}{r}$ of real numbers in the
set $\{-1,+1\}$. 

The \emph{real Grothendieck constant}, denoted $K_G$, is defined as
$\lim_{n\to\infty}K_G(n)$.  \label{def:1}
\end{definition}

 The tightest version of the
inequality known is due to Krivine~\cite{Krivine:79a}, who proved that $K_G
\leq \pi/(2\ln (1 + \sqrt 2))\approx 1.78.$ 
Davie~\cite{davie84:_lower_k_g} and, independently, Reeds~\cite{reeds91} are responsible for the best lower bounds: they showed
that $K_G \gtrsim 1.68$. The exact value of $K_G$ is unknown.

In this paper, we give a new generalization of Grothendieck's
inequality. We replace the maximization over scalars on the
right-hand side of Eq.~\eqref{gineq} with a maximization over real unit vectors of dimension
$m<n$. More formally:

\begin{definition}
Let $m$ and $n$
be positive integers with $m<n$. Let $K_G(n \mapsto m)$ be
the smallest real number such that: For all positive
integers $r$ and for all 
real $r
\times r$ 
matrices $M = (M_{ij})$, the inequality
\begin{align}\label{gengroth}
  \max_{\substack{\vsequence{a}{r}\\\vsequence{b}{r}}}\,
\sum_{i,j} M_{ij} {\vec a_i \cdot \vec b_j}\leq
\kgnm
  \max_{\substack{\vsequence{a'}{r}\\\vsequence{b'}{r}}}\, \sum_{i,j} M_{ij} {\vec
a'_i \cdot \vec b'_j}
\end{align}
holds, where the maximum on the left-hand side is taken over all 
sequences $\vsequence{a}{r},\vsequence{b}{r}$ of $n$-dimensional real unit
vectors, and the maximum on the right-hand side is taken over
all sequences
$\vsequence{a'}{r},\vsequence{b'}{r}$ of $m$-dimensional real unit
vectors.\label{def:2} This generalizes Definition~\ref{def:1} in the sense that $K_G(n) = K_G(n\mapsto 1)$.
\end{definition}

Building on the techniques Grothendieck used to prove the original
lower bound on $K_G$~\cite{grothendieck53:_resum}, we prove the following lower
bound on $K_G(n\mapsto m)$.
\begin{theorem}\label{vgrothlb} 
For all $m < n$, 
\begin{align}
  \label{eq:1}
K_G(n\mapsto  m) &\geq
\frac{m}{n}\left(\frac{\Gamma(\frac{m}{2})}{\Gamma(\frac{m+1}{2})}\frac{\Gamma(\frac{n+1}{2})}{\Gamma(\frac{n}{2})}\right)^2\\
&= 1 + \frac{1}{2m} -
\frac{1}{2n}- O(\frac{1}{m^2}).  
\end{align}
\end{theorem}
We do not need an upper bound on $K_G(n \mapsto m)$ for our quantum
application, so we don't prove one. Note however that $K_G(n \mapsto m)
\leq K_G(n)$, which does establishes a trivial upper bound. A better
upper bound could be obtained by using the techniques
in~\cite{Krivine:79a}.

\paragraph{Application to nonlocal XOR games.}
As a corollary of Theorem~\ref{vgrothlb}, we show that there are
nonlocal quantum correlations that require entangled states
with local support on a Hilbert space of dimension at least~$d$ (we
allow arbitary shared randomness). This resolves a question of Brunner
et al.~\cite{brunner:210503}, proving that what they term
\emph{dimension witnesses} exist with binary outcomes.

Brunner et al. pointed out that the same result would follow if one
could prove that the Grothendieck constants $K_G(n)$ are strictly
increasing in $n$. This is plausible but we do not know how to
prove it. Our new proof sidesteps this issue. 

\paragraph{Related work.} Definition~\ref{def:2} is but the latest in a long history of generalizations of Grothendieck's
inequality. Previously, Grothendieck's inequality has been
generalized as follows: 
\begin{itemize}
\item Replacing the real scalars, vectors and matrices with complex
ones results in our defining the \emph{complex
  Grothendieck constant}.
\item Restricting to
matrices $M$ with positive entries results in a tighter
inequality~\cite{rietz}.
\item Rather than proving inequalities that hold for all matrices, we
  can prove inequalities that only hold for all matrices $M$ of some
  fixed size, say $r \times
  s$. This refinement has been studied by Fishburn and
  Reeds~\cite{fishburn94:_bell_inequal_const_root_two}, and 
  results in the definition of a constant which they denote $K_G(r,s)$, not to be
  confused with our $K_G(n \mapsto m)$. 
\item Observe that Eq.~\eqref{gineq} has a bipartite structure, in the
  following sense: on the left-hand side, the sum is of inner
  products
  $a_i \cdot \vec b_j$ of a vector from $\vsequence{a}{r}$ with a
  vector from $\vsequence{b}{r}$; there are no inner products $a_i
  \cdot a_j$ or $b_i \cdot b_j$. A similar observation applies to the
  right-hand side. So if we consider a graph with vertices labelled by
  the vectors $a_i$ and $b_j$, and draw an edge between vertex $a_i$ and $b_j$
  whenever $M_{ij}\neq 0$, then the resulting ``interaction graph'' is
  bipartite. Alon et al. have generalized Grothendieck's
  inequality to general graphs that are not necessarily 
  bipartite~\cite{alon06:_quadr}. 
\end{itemize}

There is some earlier work on lower bounding the amount of
entanglement required to reproduce certain correlations. 
Wehner,
Christandl and Doherty show how to obtain lower bounds using information-theoretic arguments~\cite{wehner:_lower_bound_dimen_of_quant}.

The Hidden Matching quantum communication complexity problem (HM(n))~\cite{Kerenidis08} can be formulated as a nonlocal correlation, where a maximally entangled state of dimension $n$ is used to reproduce the correlations perfectly. On the other hand, using the classical bounded error one-way communication complexity lower bound for HM(n),  it  follows that one needs $\omega(\sqrt{n})$ bits of one-way communication to approximately reproduce these correlations classically. This in turn yields a lower bound on the dimension of the entangled state of $\sqrt{n}/\log n $ for any quantum strategy that  approximates these correlations. This follows because any smaller dimensional state can be used to establish a classical one-way protocol that approximates these correlations and uses less than $\omega(\sqrt{n})$ bits of communication, by simply communicating a classical description of an  approximation of  the state that Bob has after Alice did her measurement.




\paragraph{Outline.} The paper is structured as follows. We define
notation in Section~\ref{sec:notation}. In
Section~\ref{sec:an-equiv-defint}, we rework the definition of $\kgnm$
in order to work in the limit $r \to \infty$, which makes things
simpler. Then, in Section~\ref{kglbsec}, we prove our main result,
Theorem~\ref{vgrothlb}. In Section~\ref{sec:xor-games}, we describe
the consequences for quantum nonlocality. Readers wishing to skip the
details of the proof can read Section~\ref{sec:xor-games} immediately
after Section~\ref{sec:an-equiv-defint}.

\section{Notation}\label{sec:notation}
We write $[n]$ for the set~$\{1,\dots, n\}$. The unit sphere in
$\RR^n$ is denoted $S_{n-1}$. We write $\vec{da}$ for the Haar measure
on $S_{n-1}$, normalized such that $\int \vec{da} = 1$. The Dirac
delta function on $S_{n-1}$ is defined by $\delta(a-b) = 0$ if $a \neq
b$ and $\int da \delta(a-b) =1$. The norm $\|\vec a \|$ of a vector
$\vec a$ is always the Euclidean norm. In the Introduction and
Appendix, subscripts label vectors; in the remainder of the paper,
subscripts on a vector denote its componenents. Variables in
lowercase roman type will typically be vectors on the unit sphere;
variables in lowercase greek type will typically be scalars. 


\section{An equivalent defintion of $\kgnm$}
\label{sec:an-equiv-defint}
To establish a lower bound on $\kgnm$ per Eq.~\eqref{gengroth}, we need to exhibit an $r \times
r$ matrix $M$ and then calculate (or at least bound) both sides of
Eq.~\eqref{gengroth}. We will work in the limit $r \to \infty$ and so
we start by giving an alternative definition of $\kgnm$ that
facilitates this.

\begin{lemma}
\label{lemma:1}
The constant $\kgnm$ is given by 
\begin{align}\label{cgrothfrac}
K_G(n \mapsto  m) = \sup_{M:S_{n-1}\times S_{n-1}\to[-1,1]}\left(\frac{1}{D(M)}{\int
    \vec{da}\vec{db} M(\vec a, \vec b){\vec a \cdot \vec b}}{}\right).
\end{align}
where the supremum is over measurable functions
$M:S_{n-1}\times S_{n-1}\to[-1,1]$ and the denominator
\begin{align}
  \label{eq:2}
D(M) = \max_{A,B:S_{n-1}\to S_{m-1}}\int
    \vec{da}\vec{db} M(\vec a, \vec b){\vec A(\vec a) \cdot \vec B(\vec b)},
\end{align}
with the maximum over functions $A,B:S_{n-1}\to S_{m-1}$.
\end{lemma}

We informally describe why Lemma~\ref{lemma:1} is true. Fix an $r \times r$
matrix $M$, and let $\vsequence{a^*}{r},\vsequence{b^*}{r}$ be the
$n$-dimensional unit vectors that maximize \begin{align}
  \label{eq:4}
  \max_{\substack{\vsequence{a}{r}\\\vsequence{b}{r}}}\,
\sum_{i,j} M_{ij} {\vec a_i \cdot \vec b_j}, 
\end{align}
the left-hand side of Eq.~\eqref{gengroth}. Here the vectors $a^*_i$
and $b^*_j$ are labelled by indices $i$ and $j$, but these are just
dummy indices and we could have written the sum with whatever indices
we liked. The idea behind Lemma~\ref{lemma:1} is to \emph{use the
  vectors themselves as labels}, which works as long as the vectors are
all distinct. Thus we replace the matrix $M_{ij}$ in Eq.~\ref{eq:4} with
an infinite matrix $M(\vec a, \vec b)$ with rows labelled
by unit vectors $a$ and columns labelled by unit vectors $b$. The sum
over $i,j$ is replaced by integrals over $a$ and $b$. The matrix
element $M(a^*_i,b^*_j)$ is $M_{ij}$; all other entries of the matrix are
zero. 

It remains to understand what happens if two or more vectors are the
same, say $a_1^* = a_2^*$. In this case, we can replace the $r \times
r$ matrix $M$ with an $(r-1)\times r$ matrix $M'$ obtained from $M$ by
replacing the first two rows with their sum. We claim that the bound
on $\kgnm$ established by $M'$ is at least as good as the bound
established by $M$. To see this, observe that replacing $M$ with $M'$
doesn't change the value of the left-hand side of
Eq.~\ref{gengroth}. Replacing $M$ with $M'$ on the right-hand side of
Eq.~\ref{gengroth} is equivalent to performing the maximization
\begin{align}
  \label{eq:6}
  \max_{\substack{\vsequence{a'}{r}\\\vsequence{b'}{r}}}\, \sum_{i,j} M_{ij} {\vec
a'_i \cdot \vec b'_j},  
\end{align}
with the \emph{additional} constraint that $a'_1 = a'_2$, which cannot
increase the maximum. Thus the bound on $\kgnm$ obtained using $M'$ is
at least as good as that obtained using $M$. Thus it is okay to assume
that all the vectors are distinct. 

We give a formal proof of Lemma~\ref{lemma:1} in
Appendix~\ref{sec:conv-form-k_gnm}.

\section{Lower bound on $\kgnm$}\label{kglbsec}We prove Theorem~\ref{vgrothlb} by considering a specific example due
to Grothendieck himself~\cite{grothendieck53:_resum}: for $a,b\in
S_{n-1}$, take $M(\vec a, \vec b) = {\vec a\cdot \vec b}$. 

We start by calculating the denominator $D(M)$. To do this, we need to
work out which embeddings $A,B:S_{n-1} \to S_{m-1}$ achieves the maximum
in Eq.~\ref{eq:2}. It turns out that this is achieved when $A$ and $B$ are equal. Informally, we should try to preserve as much of
the structure of $S_{n-1}$ as possible, and it is natural to
conjecture that the best embedding is a projection onto an
$m$-dimensional subspace. This is indeed the case. We prove this in
the following Lemma.


\begin{lemma}\label{simpints}
For the function $M(a,b) = {a \cdot b}$, the optimal embedding
$A:S_{n-1}\to S_{m-1}$ is a projection. In particular, 
the denominator $D(M)$ is given by 
\begin{align}\label{eq:3}
D = \frac{1}{m}\biggl(\int da\, \Bigl(\sum_{i=1}^m
a_i^2\Bigr)^{1/2}\biggr)^2,
\end{align}
where $a_1, \ldots, a_n$ are the components of $a$.
\end{lemma}

\begin{proof}
We prove this result in
two steps. First, we show that the maximum is achieved by a
\emph{weighted} projection. Second, we show that the best
projection is one with uniform weights. 

We need to calculate
\begin{align}
D(M) = \max_{A,B:S_{n-1}\to S_{m-1}}\int
    \vec{da}\vec{db} M(\vec a, \vec b){\vec A(\vec a) \cdot \vec B(\vec b)},
\end{align}
with the maximum over functions $A,B:S_{n-1}\to S_{m-1}$.
For $M(a,b) = a \cdot b$, we can write
\begin{align}
  \label{eq:7}
(a\cdot
b)\big(A(a)\cdot B(b)\big) = \big(a\otimes A(a)\big)\cdot
\big(b\otimes B(b)\big),
\end{align}
(this trick is motivated by a similar one used by Krivine in proving his upper bound
on $K_G$~\cite{Krivine:79a}),
 which allows us to write $D(M)$ as a maximization over the inner product of two vectors,
\begin{align}
D= &\max_{A,B:S_{n-1}\to S_{m-1}} \left( \int da\, a \otimes A(a) \right)\cdot \left( \int db\, b \otimes B(b) \right)\\
=& \max_{A:S_{n-1}\to S_{m-1}} \left\|  \int da\, a \otimes A(a) \right\|^2,\label{absint}
\end{align}
where the second equality follows from the fact that the inner product is maximized when vectors are parallel.
Let $\int da\, a \otimes A(a) = \chi v$, where $v$
is an $(n\!+\!m)$-dimensional unit vector and $\chi \geq 0$ is what
we want to maximize. Applying the singular value
decomposition---known in quantum information theory as the Schmidt
decomposition (see for example~\cite{Nielsen:00a}---we can write
\begin{align}
  v = \sum_{i=1}^{m} \sqrt{\gamma_i} \,x_i \otimes y_i,
\end{align}
where, for each $i\in[m]$, $\gamma_i \geq 0$, $\sum_i \gamma_i = 1$,
and $\{x_1,\dots,x_m\}$ and $\{y_1,\dots,y_m\}$ are orthonormal sets
in $\RR^n$ and $\RR^m$ respectively. 
Therefore, in order to maximize 
\begin{align}\label{eq:5}
\chi  = v \cdot \int da\, a \otimes A(a) = \int da \sum_i
\sqrt{\gamma_i}(a \cdot x_i) (A(a) \cdot y_i) = \int da A(a) \cdot \Bigl(\sum_i \sqrt{\gamma_i} (a \cdot x_i) y_i \Bigr),
\end{align}
we should choose $A(a)$ to be 
\begin{align}
\frac{  \sum_i \sqrt{\gamma_i} (a \cdot x_i) y_i}{\|\sum_i
  \sqrt{\gamma_i} (a \cdot x_i) y_i\|} = \frac{  \sum_i
  \sqrt{\gamma_i}\, (a\cdot x_i) y_i}{({\sum_i \gamma_i (a \cdot x_i)^2})^{1/2}},
\end{align}
 a weighted projection onto some $m$-dimensional subspace, the particular choice of which does not matter. 
Substituting this into Eq.~(\ref{eq:5}) and then Eq.~\eqref{absint} and choosing a basis
for $\RR^n$ by extending $\vsequence{x}{m}$ so that $a_i = a \cdot
x_i$ establishes that 
\begin{align}\label{eq:3}
D = \max_{A,B:S_{n-1}\to S_{m-1}}\int dadb(a\cdot
b)\big(A(a)\cdot B(b)\big) = \bigl(\chi(\gamma_1,\ldots,
\gamma_m)\bigr)^2,
\end{align}
 where
\begin{align}
  \label{eq:8}
\chi(\gamma_1,\ldots, \gamma_m) =   \int da\, \Bigl(\sum_{i=1}^m
\gamma_i a_i^2\Bigr)^{1/2}.
\end{align}

It remains to show that weights $\gamma_i$ can be
taken to be equal. To prove this, suppose that $\chi$ is maximized by $(\gamma_1^*,
\gamma_2^*, \ldots, \gamma_m^*)$.  Then, by symmetry, the
maximum is also achieved by $(\gamma_2^*, \gamma_1^*, \ldots,
\gamma_m^*)$, and indeed, by any other permutation $\sigma$ of the
$\gamma_i^*$. Hence
\begin{align}
  \label{eq:9}
\chi(\gamma_1^*,\ldots,\gamma_m^*) 
&= \frac{1}{m!} 
\sum_{\sigma} \chi(\gamma_{\sigma(1)}^*, \ldots, \gamma_{\sigma(m)}^*)\\
&= \frac{1}{m!} \sum_{\sigma} \int da \Bigl(\sum_{i=1}^m
\gamma_{\sigma(i)} a_i^2\Bigr)^{1/2}\\
&= \int da \frac{1}{m!} \sum_{\sigma} \Bigl(\sum_{i=1}^m
\gamma_{\sigma(i)}^* a_i^2\Bigr)^{1/2}\\
&\leq \int da \Bigl(\frac{1}{m!}  \sum_{\sigma} \sum_{i=1}^m
\gamma_{\sigma(i)}^* a_i^2\Bigr)^{1/2}
\end{align}
by Jensen's inequality and the concavity of $(\cdot)^{1/2}$. But 
the coefficient of $a_i^2$ in this expression is just 
\begin{align}
  \label{eq:10}
\frac{1}{m!} \sum_\sigma \gamma_{\sigma(i)}^* = \frac{1}m \sum_i
\gamma_i^* = \frac{1}m \times 1 = \frac{1}m. 
\end{align}
Thus the maximum is achieved by uniform weights. 
 \end{proof}

With Lemma~\ref{simpints} in hand, the proof of Theorem~\ref{vgrothlb}
is straightforward.

\begin{proof}[ of Theorem \ref{vgrothlb}]
Take $M(a,b) = a \cdot b$ in Lemma~\ref{cgrothfrac}. It follows from
Lemma~\ref{simpints}, that 
\begin{align}
  \label{eq:11}
  K_G(n\mapsto m) \geq \frac{m}{n} \biggl(\frac{Y_n}{Y_m}\biggr)^2,
\end{align}
where 
\begin{align*}
Y_k : =\int_{a\in S_{n-1}}da\Bigl(\sum_{i=1}^ka_i^2\Bigr)^{1/2},
\end{align*}
and we evaluated the numerator in Eq.~\eqref{eq:3} by observing that
it is the same as the denominator when $m = n$, and so we already
calculated it as a special case of Lemma~\ref{simpints}. 
We can evaluate $Y_k$ using a trick similar to that used to calculate the surface area of the $n$-sphere. Define
\begin{align*}
C_k := \int_{a\in\RR^n}da\left(\sum_{i=1}^ka_i^2\right)^{1/2}e^{-\|a\|_2^2}.
\end{align*}
Introducing spherical coordinates, and writing $r = \|a\|_2$, we have 
\begin{align}\label{Knminusone2}
C_k = Y_k\int_0^{\infty}drr^{n-1} (r^2)^{1/2}e^{-r^2} = \frac{1}{2} Y_k\Gamma\big(\tfrac{n+1}{2}\big),
\end{align}
where $\Gamma$ is the well-known gamma function.

On the other hand, we have
\begin{align}\label{kmtwo}
C_k= \int_{-\infty}^{\infty}da_1\cdots da_{k}\left(\sum_{i=1}^ka_i^2\right)^{1/2}e^{-(a_1^2 + \cdots + a_k^2)}\int_{-\infty}^{\infty}da_{k+1}\cdots da_ne^{-(a_{k+1}^2+ \cdots + a_n)^2}.
\end{align}
We can interpret $\left(\sum_{i=1}^ka_i^2\right)^{1/2}$ as the norm of
a point in $k$-dimensional space, and write the first integral (over $k$ variables) as
\begin{align*}
\Omega_k\int_{0}^{\infty}dr' r'^ke^{-r'^2} = \frac{2\pi^{k/2}}{\Gamma(\frac{k}{2})}\cdot\frac{\Gamma\big(\tfrac{k+1}{2}\big)}{2},
\end{align*}
where $\Omega_k$ is the surface area of a unit sphere in $k$
dimensions. The second integral of Eq.~(\ref{kmtwo}) is simply
$\big(\sqrt{\pi}\big)^{n-k}$. Comparing these two ways to evaluate
$C_k$, we conclude that 

\begin{align*}
Y_k = 2\pi^{n/2}\frac{\Gamma(\frac{k+1}{2})}{\Gamma(\frac{k}{2})\Gamma(\frac{n+1}{2})}
\end{align*}
and
\begin{align}
K_G(n\mapsto m) &\geq
\frac{m}{n}\left(\frac{Y_n}{Y_{m}}\right)^2\\
&=
\frac{m}{n}\left(\frac{\Gamma(\frac{m}{2})}{\Gamma(\frac{m+1}{2})}\frac{\Gamma(\frac{n+1}{2})}{\Gamma(\frac{n}{2})}\right)^2.
\end{align}

For all integers $1\leq m<n$, this bound is nontrivial, i.e., is strictly greater
than $1$. This is because the function 
  \begin{align}
    \label{eq:13}
f(n) = \frac{1}{\sqrt{n}}\frac{\Gamma(\frac{n+1}{2})}{\Gamma(\frac{n}{2})}    
  \end{align}
is strictly increasing for $n = 1, 2, \ldots $ (see
Appendix~\ref{app:b} for a proof). Asymptotically, we have
\begin{align}
  \label{eq:12}
\kgnm
\geq 1 + \frac{1}{2m} - \frac{1}{2n} - O(\frac{1}{m^2}), 
\end{align}
where the approximation follows from the asymptotic series (see answer to Exercise 9.60 in~\cite{gkp:concrete})
\begin{align}\label{gammaasym}
\frac{\Gamma(k+ \frac{1}{2})}{\Gamma(k)} = \sqrt{k}(1 - \frac{1}{8k} + \frac{1}{128 k^2} + \dots).
\end{align}
\end{proof}

\section{Quantum nonlocality}
\label{sec:xor-games}
Here we describe the application to quantum nonlocality. Suppose that
two parties, Alice and Bob, each have a $d$-dimensional quantum
system, described by Hilbert spaces $\cH_A \cong \CC^d$ and $\cH_B
\cong \CC^d$, respectively. Alice and Bob each make a two-outcome 
measurement on their own system, resulting in outcomes
$\alpha,\beta\in\pmset{}$, respectively. Suppose the set of Alice's 
possible measurements is $M_A$, and the set of Bob's possible
measurements is $M_B$. An observable is a Hermitian operator with eigenvalues
in $\{\pm 1\}$. Alice's $a$th possible measurement is specified by an
observable $A_a$ on $\cH_A$; Bob's $b$th measurement by an observable $B_b$ on
$\cH_B$ (and all observables specify valid measurements). If the joint
system of Alice and Bob is in pure state 
$\ket \psi \in \CC^d \otimes \CC^d$, then the \emph{joint correlation}---the expectation of the product of Alice and Bob's
outcomes, given that Alice performs measurement $a$ and Bob
measurement $b$---is 
\begin{align}
  \label{eq:16}
  E[ \alpha \beta |a b] = \bra \psi A_a \otimes B_b
  \ket \psi.
\end{align}
In the computer science literature, such correlations are studied
in the context of \emph{XOR nonlocal games}~\cite{Cleve:04a}.

We say that a set of joint correlations, $\{E[ \alpha \beta |a
b]: a\in [a_{\max}], b \in [b_{\max}]\}$,
is \emph{pure-$d$-quantum-realizable} if there is a state $\ket \psi \in \CC^d
\otimes \CC^d$ and for all $a\in [a_{\max}]$, there are
observables $A_a$ on $\CC^d$ and for all $b \in [b_{max}]$, there are
observables $B_b$ on
$\CC^d$ such that $
E[\alpha \beta |a b] = \bra \psi A_a \otimes B_b
\ket \psi$. A set of joint correlations is \emph{$d$-quantum-realizable}
if it is a probabilistic mixture of pure-$d$-quantum-realizable
correlations (this definition accounts for allowing Alice and Bob to share an arbitrary large amounts of shared
randomness, use POVMs, and share a mixed state).
A set of joint correlations is
\emph{finitely quantum-realizable} if there is some $d$ such that the
correlations are $d$-quantum-realizable. 

We prove the following theorem.

\begin{theorem}\label{entxor}
For any $d$, there are correlations that are finitely quantum-realizable, but which
are not $d$-quantum-realizable.
\end{theorem}

We now describe the correlations that we use to prove Theorem~\ref{entxor}. Fix some integer
$n$. Alice and Bob's possible measurements are parametrized by unit
vectors in $\RR^n$, $a$ and $b$, respectively. (Note that each party
here has an infinite number of possible measurements; we'll reprove
the theorem with finite sets of measurements in the next subsection.) 
The joint correlations are given by 
\begin{align}
  \label{eq:18}
  E[ \alpha \beta |a b]  = a\cdot b,
\end{align}
where $a \cdot b$ is just the Euclidian inner product of $a$ and
$b$. For all $n$, these correlations are finitely quantum-realizable,
as the following result shows.

\begin{lemma}[Tsirelson~\cite{Tsirelson:85b}]\label{tsirelson}
Let $\ket \psi$ be a
maximally entangled state on $\CC^d \otimes \CC^d$ where $d =
2^{\lfloor n/2\rfloor}$. Then there are two mappings from unit vectors
in $\RR^n$ to observables on $\CC^d$, one taking $a$ to $A_a$, the
other taking $b$ to $B_b$, such that \begin{align*}
\psibra A_a\otimes B_b\psiket = a\cdot b,
\end{align*}
for all unit vectors $a, b$.
\end{lemma}

To show they are not $d$-quantum-realizable, we will use the following
characterization.

\begin{lemma}[\cite{Tsirelson:85b,toner:groth1}]\label{acin}
Suppose Alice and Bob measure observables $A_a$ and $B_b$ on a pure
quantum state $\psiket\in\CC^d\otimes\CC^d$. Then we can associate a
real unit vector $A(a)\in\RR^{2d^2}$ with $A_a$ (independent of
$B_b$), and a real unit vector $B(b)\in\RR^{2d^2}$ with $B_b$
(independent of $A_a$) such that 
\begin{align}
  \label{eq:19}
E[\alpha\beta|ab] = \bra \psi A_a \otimes B_b \ket \psi = A(a) \cdot B(b).
\end{align}
\end{lemma}

\begin{proof}[ of Theorem~\ref{entxor}]
Let $n = 2d^2 + 1$, and consider the joint correlations described in
Eq.~\ref{eq:18}. By Lemma~\ref{tsirelson}, these correlations
are finitely quantum-realizable. To show they are not
$d$-quantum-realizable, we will show that they lie outside the convex
hull of the set of pure-$d$-quantum-realizable
correlations. 

We do this in the standard way, using a Bell inequality. Consider the
following linear function on the correlations.
\begin{align}
  \label{eq:21}
  B(E[ \alpha \beta | a b ]) = \int da db (a \cdot b) E[ \alpha \beta | a b ],
\end{align}
where the integral is over all unit vectors $a$, $b$.
Substituing for $E[ \alpha \beta | a b ]$ using Eq.~\ref{eq:18}, we have 
\begin{align}
  \label{eq:22}
  B(E[ \alpha \beta | a b ]) = \int da db (a \cdot b)^2.
\end{align}
For any pure-$d$-quantum-realizable correlations, by Lemma~\ref{acin}
there are vectors $A(a)$ and $B(b)$ in $\RR^{2d^2}$, such that the
resulting correlations are given by 
\begin{align}
  \label{eq:23}
  E[ \alpha \beta |a b ]_d = A(a) \cdot B(b).
\end{align}
Evaluating $B$ on these correlations, we must have
\begin{align}
  \label{eq:24}
B(E[ \alpha \beta | a b ]_d) &= \int da db (a \cdot b)
A(a) \cdot B(b)  \\
&\leq \max_{A,B} \int da db (a \cdot b)
A(a) \cdot B(b) \\
&\leq \frac{1}{ K^{\leq}_G(n \mapsto 2d^2)}    \int da db (a \cdot
b)^2\\
& = \frac{1}{ K^{\leq}_G(n \mapsto 2d^2)}  B(E[\alpha \beta|ab])\label{eq:28}.
\end{align}
where $K^{\leq}_G(n \mapsto 2d^2)$ is our lower bound on $K_G(n \mapsto
2 d^2)$. Since $n = 2d^2 + 1$ and
$K^{\leq}_G(n \mapsto n-1)>1$ by Theorem~\ref{vgrothlb}, we conclude that the correlations in
Eq.~\ref{eq:18} are not $d$-quantum-realizable.
\end{proof}

\subsection{Reducing the number of questions}
A possible objection to the example above is that the number of
questions is taken to be infinite. Here we reduce to a finite number
of questions by considering a discretization of the unit $n$-sphere by means of an $\eps$-net.

\begin{definition}[$\eps$-net]
For fixed $\eps>0$,  a set of vectors $E_n^{\eps} = \{w_1,w_2,\dots\in S_{n-1}\}$ is an $\eps$-net for $S_{n-1}$ if for all $a\in S_{n-1}$, there exists a vector $u\in E$ that satisfies $\|a-u\|_2\leq \eps$.
\end{definition}

\begin{lemma}
For $0 < \eps <1$, there is an $\eps$-net for $S_{n-1}$ with $|E_n^{\eps}| = (3/{\eps})^n$.
\end{lemma}

\begin{proof}We follow~\cite[Lemma II.4]{hayden04:_random_quant_states}. Let $E_n^\eps$ be a
  maximal set of vectors satisfying $\|u - v\|_2 \geq \eps$ for all
  $u,v\in E_n^\eps$, where the existence of such a set is guaranteed by
  Zorn's lemma. Then $E_n^\eps$ is an $\eps$-net for $S_{n-1}$. We
  bound $|E_n^\eps|$ using a volume argument. The open balls of radius
  $\eps/2$ around each point $u \in E_n^\eps$ are pairwise disjoint
  and all contained in the ball of radius $1+\eps/2$ about the
  origin. Hence 
  \begin{align}
    \label{eq:17}
    |E_n^\eps| \leq \frac{(1+\eps/2)^n}{(\eps/2)^n} =
    \biggl(\frac{2}{\epsilon} + 1\biggr)^n \leq \biggl( \frac{3}{\eps} \biggr)^n.
  \end{align}
\end{proof}

To convert the quantum correlations above (Eq.~\ref{eq:18}) into
ones with only a finite number of settings, fix $0 < \eps < 1$ (to be
chosen later) and let $E_n^\eps$ be an $\eps$-net for $S_{n-1}$. We
shall consider the following correlations. Alice's set of possible
measurements is $E_n^\eps$, and so is Bob's (note that we implicitly apply Lemma~\ref{tsirelson} here). If Alice performs a
measurement $u \in E_n^\eps$ and Bob a measurement $v \in E_n^\eps$, the joint correlation
should satisfy
\begin{align}
  \label{eq:20}
  E[\alpha \beta | uv] = u \cdot v,
\end{align}
just as in our earlier example. These correlations, being a subset of
those considered above, are finitely quantum-realizable.

The $\eps$-net divides the unit sphere into $|E_n^\eps|$ regions. (For
$u \in E_n^\eps$, let $R_u$ be the set of points on $S_{n-1}$ that are
closer to $u$ than to any other point in $E_n^\eps$, and assign points
equidistant to two or more points in the net in some arbitrary way.) 
Consider the Bell inequality
\begin{align}
  \label{eq:25}
B_{\text{finite}}(E[\alpha \beta|uv]) = \sum_{u \in E_n^\eps} \int_{a\in R_u} da  \sum_{v \in
    E_n^\eps}\int_{b \in R_v} db \, (a \cdot b) E[\alpha \beta|uv].
\end{align}
Evaluating this on the correlations $E[\alpha\beta|uv] = u \cdot v$,
we obtain
\begin{align}
  \label{eq:25}
B_{\text{finite}}(E[\alpha \beta|uv]) &= \sum_{u \in E_n^\eps} \int_{a\in R_u} da  \sum_{v \in
    E_n^\eps}\int_{b \in R_v} db \, (a \cdot b) (u \cdot v) 
\\ & \geq - 2\eps + \int da \int db (a \cdot b)^2\\
& = B(E[\alpha \beta | a b]) - 2\eps,
\end{align}
where we used
\begin{align}
  \label{eq:26}
u \cdot v &= a\cdot b + (u-a)\cdot b + u \cdot (v-b)\\ 
&\geq a \cdot b - \|u-a\|_2 - \|v-b\|_2 \\ &\geq a \cdot b - 2 \eps,
\end{align}
and related the value of the Bell inequality to the one earlier in
this section with an infinite number of questions.

Now consider a pure $d$-dimensional quantum strategy. Let $A(u)$ be the
$2d^2$-dimensional real unit vector associated with Alice's
measurement $u$ and $B(v)$ be the vector associated with Bob's
measurement $v$ by Lemma~\ref{acin}. 
This mapping induces a mapping for the correlations
where we had an infinite number of questions. First map $a$ to the
closest point $u$ in the $\eps$-net, then to the vector $A(u)$. 
We now evaluate the Bell
inequality with an infinite number of settings in terms of the one
with a finite number:
\begin{align}
  \label{eq:27}
B(E[\alpha \beta|ab]_d) &\geq \sum_{u \in E_n^\eps} \int_{a\in R_u} da  \sum_{v \in
    E_n^\eps}\int_{b \in R_v} db \, (a \cdot b) A(u) \cdot B(v)\\
&= \sum_{u \in E_n^\eps} \int_{a\in R_u} da  \sum_{v \in
    E_n^\eps}\int_{b \in R_v} db \, (a \cdot b) E[\alpha \beta|uv]\\
&=B_{\text{finite}}(E[\alpha \beta|uv]_d).
\end{align}
Combining the above calculations with Eq.~\ref{eq:28}, we get
\begin{align}
  \label{eq:29}
  B_{\text{finite}}(E[\alpha \beta|uv]_d) &\leq B(E[\alpha \beta|ab]_d)\\
  &\leq \frac{1}{ K^{\leq}_G(n \mapsto 2d^2)}  B(E[\alpha \beta|ab])\\
 & \leq \frac{1}{ K^{\leq}_G(n \mapsto 2d^2)} \left[ B_{\text{finite}}(E[\alpha \beta|ab]) +
  2\eps \right]
\end{align}

Now choose $n > 2d^2$ and $\eps$ such that $B_{\text{finite}}(E[\alpha
\beta |u v]_d)$ is strictly less than $B_{\text{finite}}(E[\alpha
\beta |u v])$. 

\section*{Acknowledgements}
We thank Stephanie Wehner for sharing with us a draft of
Ref.~\cite{wehner:_lower_bound_dimen_of_quant}.


\appendix
\section{Convenient form of $K_G(n\mapsto m)$}
\label{sec:conv-form-k_gnm}
Here we prove Lemma~\ref{lemma:1}. Observe that we can rewrite the
conventional definition of $\kgnm$ as 
\begin{align*}
K_G(n\mapsto m) =
\lim_{r\to\infty}\sup_{M_{ij}}\left(\frac{\max_{(a_i,b_j)}\sum_{i,j}M_{ij}{\vec
      a_i \cdot \vec b_j}}{\max_{(a_i',b_j')}\sum_{i,j}M_{ij}a_i'\cdot b_j'}\right).
\end{align*}
The following two propositions give the result.

\begin{proposition}\label{rgeql}
For all positive integers $r$, and for all
$r \times r$ real-valued matrices $M_{ij}$, there exists a measurable function
$M':S_{n-1}\times S_{n-1} \to [-1,1]$, such that
\begin{align}\label{rgeqleq}
\frac{\int \vec{da}\,\vec{db}M'(a,b)a\cdot b}{\max_{A,B:S_{n-1}\to
    S_{m-1}}\int \vec{da}\,\vec{db}M'(a,b)A(a)\cdot B(b)} \geq
\frac{\max_{(a_i,b_j)}\sum_{i,j}M_{ij}{\vec a_i \cdot \vec
    b_j}}{\max_{(a_i',b_j')}\sum_{i,j}M_{ij}a_i'\cdot b_j'}.
\end{align}
\end{proposition}

\begin{proof}
Let $f,g: [r]\to S_{n-1}$ and $f',g':[r]\to S_{m-1}$ be vector valued
functions, and let $f^*$ and $g^*$ be such that they give a sequence
$(f^*(i),g^*(j))_{i,j=1}^r = (a_i^*,b_j^*)_{i=1}^r$ that maximizes
$\sum_{i,j}M_{ij}{\vec a_i \cdot \vec b_j}$. Set
$$
M'(\vec a, \vec b) = \sum_{i,j}M_{ij}\delta(a - a_i^*)\delta(b-b_j^*),
$$
where $\delta(\cdot)$ denotes the Dirac delta function. This causes the numerators of (\ref{rgeqleq}) to be equal. For the denominator of left-hand side of (\ref{rgeqleq}), we have
\begin{align}
\max_{A,B:S_{n-1}\to S_{m-1}}\int dadbM'(a,b)A(a)\cdot B(b) &=& \max_{A',B':(f^*(i),g^*(j))\to S_{m-1}}\sum_{i,j}M_{ij}A'(f^*(i))\cdot B'(g^*(j))\\
&\leq& \max_{f',g'}\sum_{i,j}M_{ij}f'(i)\cdot g'(j),
\end{align}
where the inequality follows because because the second maximization is over a subset of the set that the last equation is maximized over. This gives the result.
\end{proof}

\begin{proposition}\label{lgeqr}
For any measurable function $M'(a,b)$ with $a,b\in S_{n-1}$, and any $\eps> 0$, there exist an $r$ and matrix $M_{ij}\in\RR^r\times\RR^r$, such that
\begin{align}\label{lgeqreq}
 \frac{\max_{(a_i,b_j)}\sum_{i,j}M_{ij}{\vec a_i \cdot \vec b_j}}{\max_{(a_i',b_j')}\sum_{i,j}M_{ij}a_i'\cdot b_j'} \geq \frac{\int \vec{da}\,\vec{db}M'(a,b)a\cdot b}{\max_{A,B:\RR^n\to\RR^m}\int \vec{da}\,\vec{db}M'(a,b)A(a)\cdot B(b)} - \eps.
\end{align}
\end{proposition}

\begin{proof}
First note that since $|M'(\cdot,\cdot)|$ is a measurable function, the integral $\int dadb|M'(a,b)|$ is bounded. Therefore, without loss of generality, we may assume that $\int dadb|M'(a,b)|=1$.

Suppose we divide the unit $n$-sphere up into $r$ disjoint regions $R_1,\dots,R_r\subseteq S_{n-1}$ whose sizes decrease with increasing $r$, and set
\begin{align*}
M_{ij} = \int_{a\in R_i}\int_{b\in R_j}dadbM'(a,b).
\end{align*}
Let $(a_i^*,b_j^*)_{i,j=1}^r$ be a sequence that maximizes
$\sum_{i,j}M_{ij}{\vec a_i \cdot \vec b_j}$, and define $\delta := \max_{i,j}\{|a_i^*\cdot b_j^* - a\cdot b|\mid a\in R_i,b\in R_j\}$. Then  by the triangle inequality, we have
\begin{align*}
\left|\sum_{i,j}\int_{a\in R_i}da\int_{b\in R_j}dbM(\vec a, \vec b)\left(a_i^*\cdot b_j^* - a\cdot b\right)\right| \leq \sum_{i,j}\int_{a\in R_i}da\int_{b\in R_j}db|M'(a,b)||a_i^*\cdot b_j^* - a\cdot b| \leq \delta.
\end{align*}
Hence, for the numerators we get:
$\max_{(a_i,b_j)}\sum_{i,j}M_{ij}{\vec a_i \cdot \vec b_j}\geq \int \vec{da}\,\vec{db}M'(a,b)a\cdot b - \delta$. For the denominators, we have
\begin{align*}
\max_{A,B:\RR^n\to\RR^m}\int \vec{da}\,\vec{db}M'(a,b)A(a)\cdot B(b) \geq \max_{(a_i',b_j')}\sum_{i,j}M_{ij}a_i'\cdot b_j',
\end{align*}
since  we can always pick $A(a) = a_i'$ and $B(b) = b_j'$ for all $a\in R_i$ and $b\in R_j$. The result follows from the fact that we can let $\delta$ become arbitrarily small by increasing $r$.
\end{proof}

\section{Proof that the bound on $\kgnm$ is nontrivial}
\label{app:b}
Here we establish the following lemma.
\begin{lemma}\label{lem:function}
  The function 
  \begin{align}
    \label{eq:13}
f(n) = \frac{1}{\sqrt{n}}\frac{\Gamma(\frac{n+1}{2})}{\Gamma(\frac{n}{2})}    
  \end{align}
is strictly increasing on integers $n = 1,2,\ldots$. 
\end{lemma}

\begin{proof}
  For $n\leq 9$, just evaluate $f(n)$. For $n>9$, we use the following bound on
  $\log \Gamma(x)$, first proved by
  Robbins~\cite{robbins55:_remar_of_stirl_formul} for
  integer values of $x$, but which Matsunawa observed~\cite[Remark 4.1]{matsunawa76:_some_inequal_based_inver_factor_series} is also valid for
  real values of $x\geq 2$:
\begin{align}
  \label{eq:15}
 \sqrt{2\pi} x^{x+1/2} e^{-x+1/(12x+1)} < \Gamma(x+1) <  \sqrt{2\pi} x^{x+1/2} e^{-x+1/(12x)}.
\end{align}
Using this bound, we obtain
\begin{align*}
\log \frac{f(n+1)}{f(n)} 
& = -\frac{1}{2} \log \Bigl( 1 + \frac{1}{n} \Bigr)  + \log
\frac{n}{2}+2\log \Gamma(\frac{n}{2}) -2 \log
\Gamma(\frac{{n+1}}{2})\\
& 
\geq -\frac{1}{2} \log \Bigl( 1 + \frac{1}{n} \Bigr) 
+  \log \Bigl( 1 + \frac{1}{n/2-1} \Bigr)
- n  \log \Bigl( 1 + \frac{1}{n-2} \Bigr)
+ \frac{2}{6n-11} - \frac{2}{6n-6}.
\end{align*}

Now use 
\begin{align}
  \label{eq:14}
\frac{1}{n} - \frac1{2n^2} + \frac1{3n^3} -
\frac1{4n^4} \leq 
\log \Bigl(1+\frac1{n}\Bigr)\leq \frac{1}{n} - \frac1{2n^2} + \frac1{3n^3},
\end{align}
(which is valid for all $n\geq 1$), and we obtain
\begin{align*}
\log \frac{f(m+10)}{f(m+9)} 
\geq \frac{14 m^7+679 m^6+13923 m^5+155346 m^4+1005620 m^3+3684139 m^2+6679947 m+3828140}{12 (m+7)^4 (m+8)
   (m+9)^3 (6 m+43)},
\end{align*}
which is obviously positive when $m\geq 0$, i.e., when $n\geq 9$. Thus $f(n)$ is strictly increasing. 
\end{proof}

\end{document}